\documentclass[letterpaper]{article} 
\usepackage{aaai2026}  
\usepackage{times}  
\usepackage{helvet}  
\usepackage{courier}  
\usepackage[hyphens]{url}  
\usepackage{graphicx} 
\usepackage{natbib}  
\usepackage{caption} 
\usepackage{algorithm, algpseudocode}
\usepackage{mathtools}
\usepackage{amsfonts,bm}
\usepackage{amsmath, amssymb, amsthm}
\usepackage{dsfont}

\usepackage{newfloat}
\usepackage{listings}
\DeclareCaptionStyle{ruled}{labelfont=normalfont,labelsep=colon,strut=off} 
\floatstyle{ruled}
\newfloat{listing}{tb}{lst}{}
\floatname{listing}{Listing}
\def\gO{{\mathcal{O}}}
\def\gN{{\mathcal{N}}}
\def\gA{{\mathcal{A}}}
\def\mU{{\bm{U}}}
\def\gS{{\mathcal{S}}}
\def\gF{{\mathcal{F}}}
\def\gH{{\mathcal{H}}}
\def\gP{{\mathcal{P}}}
\def\gB{{\mathcal{B}}}
\def\sP{{\mathbb{P}}}
\def\sR{{\mathbb{R}}}

\DeclarePairedDelimiter\brc{\{}{\}}
\DeclarePairedDelimiter\abs{\lvert}{\rvert}
\newcommand{\ind}[1]{\mathds{1}\brc*{#1}}
\DeclareMathOperator*{\argmax}{arg\,max}

\newtheorem{thm}{Theorem}
\newtheorem{defn}{Definition}
\newtheorem{lem}{Lemma}
\newtheorem{rema}{Remark}
\newtheorem{coro}{Corollary}
\title{Bandit Learning in Housing Markets}
\author{
     Shiyun Lin
}
\affiliations{
    Center for Statistical Science, School of Mathematical Sciences\\
    Peking University\\
    shiyunlin@stu.pku.edu.cn
}

\usepackage{bibentry}

\begin{document}

\maketitle

\begin{abstract}
The housing market, also known as one-sided matching market, is a classic exchange economy model where each agent on the demand side initially owns an indivisible good (a house) and has a personal preference over all goods. The goal is to find a core-stable allocation that exhausts all mutually beneficial exchanges among subgroups of agents. While this model has been extensively studied in economics and computer science due to its broad applications, little attention has been paid to settings where preferences are unknown and must be learned through repeated interactions.
In this paper, we propose a statistical learning model within the multi-player multi-armed bandit framework, where players (agents) learn their preferences over arms (goods) from stochastic rewards. We introduce the notion of \emph{core regret} for each player as the market objective. We study both centralized and decentralized approaches, proving $\mathcal{O}(\log T / \Delta^2)$ upper bounds on regret, where $T$ is the time horizon and $\Delta$ is the minimum preference gap among players. For the decentralized setting, we also establish a matching lower bound, demonstrating that our algorithm is order-optimal. \looseness=-1
\end{abstract}


\section{Introduction}\label{sec:intro}
The housing market, also known as one-sided matching market, represents a foundational economic institution where agents exchange property rights to achieve more desirable arrangements. Classical research in this domain traces back to the seminal work of~\citet{shapley1974cores}, and the model has been extensively studied in the literature~\citep{roth2023online} due to its wide range of applications such as student housing assignments~\citep{sonmez2010house}, school choice~\citep{abdulkadirouglu2003school} and kidney exchange~\citep{roth2004kidney}. In a housing market, each agent starts with an indivisible good -- typically representing a house -- and has their own personal preferences over all available goods in the market. Unlike traditional markets with buyers and sellers, here, agents trade their initial endowments based on their preferences, seeking to obtain a house they value more, and the houses must be traded without monetary transactions. The key challenge is to find a stable and fair allocation where no agent can improve their outcome by trading with someone else. 
The Top Trading Cycle (TTC) algorithm introduced by \citet{shapley1974cores} is a celebrated mechanism that ensures efficient, strategy-proof and Pareto optimal outcomes~\citep{roth1977weak}. Subsequent research has expanded this framework to accommodate different settings and provide various solutions to the problem~\citep{hylland1979efficient,echenique2021constrained,garg2024one}. 
Despite these advancements, the integration of adaptive learning mechanisms into  housing market models remains an underexplored area, particularly in online environments where participants must navigate uncertainty in preferences. 

The assumptions of a known, full preference profile in housing markets is often unrealistic. In practice, participants -- such as tenants on LeaseSwap NYC, patients in kidney exchange, or traders in NFT markets -- frequently lack well-defined preferences over items they have not experienced. Fortunately, these settings typically allow for repeated short-term interactions that yield immediate feedback, for example, through home visits, medical compatibility tests, or temporary NFT licensing. This enables agents to learn their uncertain preferences through iterative matchings, circumventing the limitations of classical models.

The multi-armed bandit (MAB) framework models how a player learns in an unknown environment with limited feedback~\citep{auer2002finite, lattimore2020bandit}. In the basic setup, a player faces $K$ arms, each with an initially unknown reward distribution. Upon selecting an arm, the player observes a stochastic reward and updates their belief about the arm's preference. The goal is to maximize cumulative expected rewards, or equivalently, minimize cumulative expected regret -- the difference between rewards from the optimal arm and the player's chosen arms over time. Classical strategies for balancing exploration (learning arm preferences) and exploitation (leveraging known rewards), such as explore-then-commit (ETC)~\citep{garivier2016explore}, upper confidence bound (UCB)~\citep{auer2002finite} and Thompson sampling (TS)~\citep{thompson1933likelihood}, achieve sublinear regret, ensuring asymptotic optimality.\looseness=-1

The online learning setting in housing market mirrors the exploration-exploitation trade-off central to MAB problems, where agents sequentially select actions to maximize cumulative rewards amid uncertainty. The dynamic and competitive nature of housing markets further complicates this learning process, as agents' decisions influence not only their own outcomes but also the opportunities available to others.

To formalize these challenges, we initiate \emph{bandit learning in housing markets}, abstract the market as a multiplayer bandit problem. Here, players and arms correspond to agents and houses, each player has heterogeneous and unknown preferences over the arms, and each player is associated with an arm, which denotes their initially endowed house. When being matched to an arm, the player could learn the corresponding preference through the stochastic reward, but every arm could be matched to at most one player every time, when more than one player try to pull the same arm, all of them would get collision and receive no rewards. Taking the outcome from the TTC algorithm as a natural and desirable solution for the market, denoted as the core matching, our objective is to minimize the regret defined as the cumulative reward difference between the arm from the core matching and the player's selected arm. We develop matching and learning algorithms that can provably attain the core of the market in this setting. Our contributions are as follows:

\begin{itemize}
	\item We introduce a novel model for housing markets in which agents initially lack knowledge of their preferences over houses but can repeatedly interact with the market to gradually learn these preferences. 
	A key contribution is the definition of a natural notion of regret, grounded in the cooperative game-theoretic concept of the \emph{core}, which quantifies the exploration-exploitation trade-off faced by individual players. 
	
	\item When the horizon $T$ of the bandit problem is known, we propose an ETC-type algorithm in the decentralized market setting, and prove an $\gO(N \log T / \Delta_{\min}^2)$ problem-dependent upper bounds on the regret for every player, where $N$ is the number of players, and $\Delta_{\min}$ is the players' minimum preference gap. 
	
	\item When the horizon $T$ is unknown, we provide a UCB-type algorithm in the centralized market setting, which is adaptive and anytime. We prove that centralized UCB achieves $\gO(N^2 \log T / \Delta_{\min}^2)$ problem-dependent upper bounds on the regret for every player.
	
	\item For the decentralized setting, we prove a matching lower bound of $\Omega (N \log T / \Delta_{\min}^2)$. To establish this, we construct an instance where a single player's exploration inevitably causes collisions for others. Specifically, any algorithm must spend $\Omega(\log T / \Delta_{\min}^2)$ rounds for a player to identify its optimal arm when the preference gap is $\Delta_{\min}$. In our construction, each such exploration step by one player forces a collision upon a specific, distinct player. Consequently, if $N - 1$ players each have a minimum gap of $\Delta_{\min}$, their collective exploration imposes $\Omega(N \log T / \Delta^2_{\min})$ total collisions -- and thus regret -- upon the remaining player.
\end{itemize}

\section{Related Work}\label{sec:literature}
\paragraph{Multi-player Bandit Learning}
The multi-player bandit problem involves multiple decentralized players interacting with a shared multi-armed bandit environment. When players pull the same arm simultaneously, a collision occurs, resulting in a loss. In settings where arm means vary across players, the benchmark is typically the \emph{maximum weight matching}, and the \emph{system regret} -- defined as the cumulative reward loss summed over all players -- measures performances. \citet{tibrewal2019distributed} proposed an ETC type algorithm where players exploit the best-estimated matching, \citet{mehrabian2020practical} combined forced collisions for implicit communication with matching eliminations, \citet{shi2021heterogeneous} adapted the CUCB algorithm~\citep{chen2013combinatorial} to this setting. For a comprehensive survey, see~\citet{boursier2024survey}.

While our model shares the reward and collision structure of heterogeneous multiplayer bandits, our benchmark differs fundamentally: we use the \emph{core} of the housing market —- a game-theoretic solution distinct from maximum weight matching -- and define \emph{individual regret} per player rather than aggregate system regret. This ensures fairness but introduces distinct algorithmic challenges.

\paragraph{Competing Bandits in Two-sided Matching Markets}
Bandit problems in matching markets were first formalized by \citet{das2005two}, with subsequent works~\citep{liu2020competing,liu2021bandit,sankararaman2021dominate,kong2023player,lin2025stable} exploring this model to achieve stable matching~\citep{gale1962college}. In these two-sided markets, players (with unknown utilities) and arms (with known preferences) interact -- when multiple players pull the same arm, only the top-ranked player by the arm's preference gets matched while others face collisions, with individual regrets defined accordingly. Since multiple stable matchings may exist, regrets are typically measured against either player-optimal or player-pessimal stable matchings. \looseness=-1 

In contrast, we study housing markets (one-sided matching) where only players have preferences over arms. This creates a distinct collision structure: when multiple players propose to the same arm, all receive zero reward. Moreover, the core matching in a housing markets is unique, ensuring an unambiguous regret definition. These differences from two-sided markets introduce novel algorithmic challenges.\looseness=-1

\section{Problem Setting}\label{sec:setting}
Denote $N$ as the number of players in the market, and every player has an initial endowed arm. Let $\gN = \left\{p_1, p_2, \cdots, p_N\right\}$ be the player set and $\gA = \left\{a_1, a_2, \cdots, a_N\right\}$ be the arm set.
The preferences of players on arms can be represented through a utility matrix $\mU$, where $\mU(i, j) \in \left[0, 1\right]$ denotes the preference of player $p_i$ on arm $a_j$. If $\mU(i, j) > \mU(i, j')$, player $p_i$ prefers arm $a_j$ over $a_{j'}$, and we denote as $a_j \succ_{i} a_{j'}$. In this paper, we assume all preferences are distinct, i.e., $\mU(i, j) \neq \mU(i, j')$ for any different arms $a_j \neq a_{j'}$.
Additionally, the utility vectors can vary entirely between players, reflecting heterogeneous preferences over arms.
In each round $t = 1, 2, \cdots$, every player $p_i$ proposes to an arm $A_i(t) \in \gA$. Each arm $a_j$ then receives applications from the set of players $A_j^{-1}(t) := \left\{p_i: A_i(t) = a_j\right\}$. Since arm $a_j$ is initially owned by player $p_j$, we assume only $p_j$ observes the application profile $A_j^{-1}(t)$, while other players remain unaware of this information. Arms are indivisible goods with no preferences over players, if multiple players propose to $a_j$ in the same round, i.e., $\abs{A_j^{-1}(t)} > 1$, a collision occurs, and $a_j$ is not matched to any player. In this case, the successfully matched player for $a_j$ is $\bar{A}_j^{-1} (t) = \emptyset$, and each proposing player $p_i \in A_j^{-1}(t)$ receives $\bar{A}_i(t) = \emptyset$, along with a deterministic reward $X_i(t) = 0$ indicating they are blocked.
Conversely, if only one player $p_i$ proposes to $a_j$, i.e., $\abs{A_j^{-1}(t)} = 1$, the match succeeds: $\bar{A}_j^{-1}(t) = p_i$, and $p_i$ receives a random reward $X_i(t)$ characterizing its matching experience in this round, which we assume is 1-subgaussian with expectation $\mU(i, \bar{A}_i(t))$. The set of all successfully matched player-arm pairs in round $t$ forms a matching, denoted $\mu_t$, where $\mu_t(p_i) = \bar{A}_i(t)$.

The core is a fundamental solution concept in housing markets~\citep{shapley1974cores}, which is the set of feasible allocations where no coalition of agents can benefit by breaking away from the grand coalition. Formally,
\begin{defn}[Core]\label{defn:core}
	A matching $\mu$ is in the core if no coalition $\gS \subseteq \gN$ can block $\mu$, i.e., there does not exist an alternative matching $\mu'$ such that $\mu'(p_i) \succ_i \mu(p_i), \forall p_i \in \gS$.
\end{defn}
In housing markets with strict preferences, the core is always nonempty, and consists of a unique matching~\citep{roth1977weak}, which we refer to as the \emph{core matching} and denote by $\mu^*$. Our objective is to learn $\mu^*$ while minimizing the \emph{core regret} for each player $p_i \in \gN$. This regret is defined as the cumulative difference, over $T$ rounds, between the expected rewards from being matched to $\mu^*(p_i)$ and the rewards actually obtained by $p_i$:
{\footnotesize
\begin{align}\label{eq:core_regret}
	Reg_i(T) = T \cdot \mU(i, \mu^*(p_i)) - \mathbb{E} \left[\sum_{t = 1}^T X_i(t)\right].
\end{align}
}
The expectation is taken over the randomness of the received reward and the players' strategy.

\paragraph{Offline Top-Trading-Cycle Algorithm}
In the \emph{offline} setting where all players know their exact preferences, the Top-Trading-Cycle (TTC) algorithm~\citep{shapley1974cores} efficiently computes the unique core allocation of the housing market -- a stable assignment where no coalition of players can improve their outcomes through mutual exchanges. The centralized TTC proceeds iteratively. First, each player identifies her most preferred available arm. Second, a directed graph is formed  where players point to owners of their top choices. Third, at least one cycle (including possible self-cycles) is identified and implemented, with involved players exiting the market. This process terminates within $N$ steps, guaranteed to produce a core allocation.

\section{Decentralized Algorithm}\label{sec:decentralize}
In this section, we present an Explore-then-Commit (ETC) algorithm for decentralized housing markets, operating in two phases. First, players explore arms in a round-robin fashion to estimate their preference rankings (Line~\ref{algline:begin_explore}-\ref{algline:end_explore}), ensuring reliable ordinal estimates for matching. Next, players use these estimates to identify their assigned arms in the core matching via a decentralized adaptation of the \emph{You Request My House, I Get Your Turn (YRMH-IGYT)} mechanism~\citep{abdulkadirouglu1999house} (Line~\ref{algline:begin_ttc}-\ref{algline:end_ttc}). Once matched, players commit exclusively to these arms in all future rounds.

\begin{algorithm}[tp]
	{\small 
		\caption{Decentralized Explore-then-YRMH-IGYT (from view of player $p_i$)}\label{alg:decentralize}
		\begin{algorithmic}[1]
			\Require player set $\gP$ and arm set $\gA$, horizon $T$.
			
			\State Initialize: $\hat{\mU}(i, j) = 0$, $T_{i, j} = 0$, $\forall j \in [N]$.
			
			\State //Phase 1, learn the preferences \label{algline:begin_explore}
			\For{$\ell = 1, 2, \cdots$}
			\State $P_\ell^{(i)} = $ False //whether the preference is well-estimated
			\For{$t = \sum_{\ell' = 1}^{\ell - 1} (2^{\ell'} + N) + 1, \cdots, \sum_{\ell' = 1}^{\ell - 1} (2^{\ell'} + N) + 2^{\ell}$} \label{algline:begin_round_robin}
			\State $A_i(t) = a_{(i + t - 1) \% N + 1}$. \label{algline:round_robin}
			\State Observe $X_{i, A_i(t)}(t)$, update $\hat{\mU}(i, A_i(t)), T_{i, A_i(t)}$. \label{algline:update}
			\EndFor \label{algline:end_round_robin}
			\State Compute $\text{UCB}_{i, j}$ and $\text{LCB}_{i, j}$ for each $j \in [N]$.
			\If{$\exists \sigma$ such that $\text{LCB}_{i, \sigma_{k}} > \text{UCB}_{i, \sigma_{k + 1}}$, $\forall k \in [N]$} \label{algline:rank}
			\State $P_{\ell}^{(i)} = $ True and $\sigma_i = \sigma$.
			\EndIf
			\For{$t = \sum_{\ell' = 1}^{\ell - 1} (2^{\ell'} + N) + 2^{\ell} + 1, \cdots, \sum_{\ell' = 1}^{\ell} (2^{\ell'} + N)$} \label{algline:begin_communication}
			\State $t' = t - \sum_{\ell' = 1}^{\ell - 1} (2^{\ell'} + N) - 2^{\ell}$
			\If{$P_{\ell}^{(i)} == $ True}
			\State $A_i(t) = a_{t'}$. \label{algline:communicate_propose}
			\Else
			\State $A_i(t) = \emptyset$. \label{algline:give_up}
			\EndIf
			\If{$i == t'$ and $\abs{A_i^{-1}(t)} == N$} \label{algline:estimate_well}
			\State Enter in Phase 2 with $\sigma_i = (\sigma_{i, 1}, \cdots, \sigma_{i, N})$. \label{algline:enter_second}
			\State $t_1 = \sum_{\ell' = 1}^{\ell} (2^{\ell'} + N)$ //Phase 1 ends.
			\EndIf
			\EndFor \label{algline:end_communication}
			\EndFor \label{algline:end_explore}
			\State //Phase 2, find the core matching arm with $\sigma_i$ \label{algline:begin_ttc}
			\State $t = t_1 + 1$
			\State $F_i = $ True //Whether the initial endowed arm is still available in the market, this flag is known to all players
			\State $k = 1$, $\gA_{k} = \left\{a_i: F_i == \text{True}\right\}$.
			\State $Propose(i) = $ False. //A flag indicated whether player $p_i$ proposed to some arm in the $k$-th epoch.
			\State $i_{\min}= \min( \left\{i | a_i \in \gA_k\right\})$.
			\While{$\abs{\gA_k} > 0$ and $F_i ==$ True}
			\State $j_{\min}^{(i)} = \min\left\{j | a_{\sigma_{i, j}} \in \gA_k\right\}$. //The best available arm for player $p_i$.
			\If{$t == t_1 + 1$ or $k(t) == k(t - 1) + 1$}
			\If{$i == i_{\min}$}
			\State $A_i(t) = a_{j_{\min}^{(i)}}$. \label{algline:leader_propose}
			\State $Propose(i) = $ True. \label{algline:leader_mark}
			\EndIf
			\ElsIf{$\abs{A_i^{-1}(t - 1)} > 0$}
			\State $A_i(t) = a_{j_{\min}^{(i)}}$, $Propose(i) = $ True.
			\State $j_r = \left\{j | a_j = A_i^{-1}(t - 1)\right\}$.
			\EndIf
			\If{($Propose(i) == $ True and $\abs{A_i^{-1}(t)} > 0$) or ($F_{j_r} == $ False)} //In a cycle \label{algline:cycle_complete}
			\State $F_i =$ False, $A_i(s) = a_{j_{\min}^{(i)}}$, $\forall s > t$. \label{algline:cycle_set}
			\EndIf
			\State $\gB = \left\{a_i: F_i == True\right\}$.
			\If{$\gB \neq \gA_k$} //The available arm set is updated
			\State $k = k + 1$, $\gA_k = \gB$, $Propose(i) =$ False.
			\State $i_{\min} = \min \left\{i | a_i \in \gA_{k}\right\}$.
			\EndIf
			\EndWhile \label{algline:end_ttc}
		\end{algorithmic}
	}%
\end{algorithm}

The first phase of the Algorithm~\ref{alg:decentralize} consists of multiple sub-phases $\ell = 1, 2, \cdots$, each lasting $2^{\ell} + N$ rounds (Line~\ref{algline:begin_explore}-\ref{algline:end_explore}). Each sub-phase $\ell$ begins with an exploration stage of $2^{\ell}$ rounds (Line~\ref{algline:begin_round_robin}-\ref{algline:end_round_robin}), followed by $N$ communication rounds (Line~\ref{algline:begin_communication}-\ref{algline:end_communication}).
During the exploration stage, players aim to gather sufficient observations about the arms. The subsequent communication rounds allow them to verify whether all $N$ players have accurately learned their preferences. Once this condition is met, players exit the exploration phase and proceed to the second phase, where they identify arms in the core matching (Line~\ref{algline:enter_second}).

During the exploration stage of each sub-phase, players follow a round-robin strategy to propose to arms (Line~\ref{algline:round_robin}). Since each player initially owns a distinct endowed arm with a unique index, this ensures that no two players select the same arm simultaneously. As a result, all proposals are successfully accepted without collisions. When player $p_i$ receives an observation from the selected arm $A_i(t)$, it updates both the estimated preference value $\hat{\mU}(i, A_i(t))$ and the observation count $T_{i, A_i(t)}$ for that arm (Line~\ref{algline:update}). The update rule is as follows,
{\footnotesize
\begin{align}\label{eq:empirical_mean}
	\hat{\mU}(i, A_i(t)) &= \frac{\hat{\mU}(i, A_i(t)) \cdot T_{i, A_i(t)} + X_{i, A_i(t)}(t)}{T_{i, A_i(t)} + 1}, \\
	T_{i, A_i(t)} &= T_{i, A_i(t)} + 1.
\end{align}
}
At the end of the exploration phase, players construct a confidence set for their estimated preference values using collected observations. Specifically, for player $p_i$, the confidence interval for the preference value of arm $a_j$ is defined by an upper bound (UCB) and a lower bound (LCB), given as \looseness=-1
{\footnotesize
\begin{align}\label{eq:confidence_bound}
	\begin{split}
		\text{UCB}_{i, j} &= \hat{\mU}(i, j) + \sqrt{\frac{6 \log T}{\max \left\{T_{i, j}, 1\right\}}}, \\
		\text{LCB}_{i,j} &= \hat{\mU}(i, j) - \sqrt{\frac{6 \log T}{\max \left\{T_{i, j}, 1\right\}}}.
	\end{split}
\end{align}
}
When the confidence intervals of two arms $a_j, a_{j'}$ become disjoint, that is, when $\text{LCB}_{i, j} > \text{UCB}_{i, j'}$ or vice versa, player $p_i$ can establish its strict preference ordering between them. Once $p_i$ successfully determines a complete preference ranking over all $N$ arms (Line~\ref{algline:rank}), it sets the flag $P_{\ell}^{(i)}$ to True for the current sub-phase and records this permutation as its estimated preference ranking $\sigma_i$.

Communication rounds enable players verify if all participants have accurately estimated their preference rankings. In the $i$-th communication round of every sub-phase, player $p_i$'s endowed arm serves as a broadcast channel. A player $p_j$ proposes to $p_i$'s arm \emph{if and only if} it has an accurate ranking (Line~\ref{algline:communicate_propose}); otherwise, it abstains (Line~\ref{algline:give_up}). While this may cause collisions, the goal is not to resolve them but to convey status information. If $p_i$ observes proposals from all $N$ players including itself (Line~\ref{algline:estimate_well}), it infers universal estimation success and triggers the transition to the second phase using its estimated ranking $\sigma_i$ (Line~\ref{algline:enter_second}).

The second phase (Line~\ref{algline:begin_ttc}-\ref{algline:end_ttc}) adapts the YRMH-IGYT algorithm~\citep{abdulkadirouglu1999house} to our decentralized setting for players to progressively identify their core allocations. The process unfolds in successive sub-phases. In each, the unassigned player with the smallest index initiates a proposal chain by proposing to their most preferred remaining arm (Lin~\ref{algline:leader_propose}-\ref{algline:leader_mark}). The recipient of a proposal then proposes to their own top choice among available arms, and this sequence continues. A top trading cycle is completed when a proposal is received by a player who is already part of the current chain (Line~\ref{algline:cycle_complete}). The player who closes the cycle -- the first to receive a repeat proposal -- then marks their endowed arm as allocated. This allocation status propagates backward through the chain. All players involved in the cycle permanently fix their allocations and remove their own endowed arms from the market (Line~\ref{algline:cycle_set}). The process repeats among the remaining players until all arms are allocated, thus achieving a stable core matching in a fully decentralized manner. \looseness=-1

\begin{rema}
	Algorithm~\ref{alg:decentralize} relies on standard multiplayer bandit assumptions for implicit coordination: (1) Unique, globally-known arm IDs, which also serve as player identifiers for round-robin exploration and proposal-chain initiation. (2) A shared global clock to synchronize the transition to Phase 2. (3) Common knowledge of the available arms in each round, enabling players to reconstruct the flag $F_i$.
\end{rema}

\subsection{Theoretical Analysis}\label{subsec:theory_decentralize}
Prior to presenting the formal regret analysis of Algorithm~\ref{alg:decentralize}, we introduce the notion of preference gaps to quantify the intrinsic difficulty of the learning problem.

\begin{defn}[Minimum Preference Gap]\label{def:pref_gap}
	For each player $p_i$ and arm $a_j \neq a_{j'}$, let $\Delta_{i, j, j'} = \mU(i, j)- \mU(i, j')$ be the preference gap of $p_i$ between $a_j$ and $a_{j'}$. Let $r_i$ be the preference ranking of player $p_i$ and $r_{i, k}$ be the $k$-th preferred arm in $p_i$'s ranking for $k \in [N]$. Define $\Delta_{\min} = \min_{i \in [N]; k \in [N]} \Delta_{i, r_{i, k}, r_{i, k + 1}}$ as the minimum preference gap among all workers and their preferences over the arms. $\Delta_{\min}$ is non-negative since all preferences are distinct.
\end{defn}

We now present the upper bound for the core regret for each player by following Algorithm~\ref{alg:decentralize}.

\begin{thm}\label{thm:decentralize_upper_bound}
	Following Algorithm~\ref{alg:decentralize}, the core regret of each player $p_i \in \gN$ satisfies
	{\scriptsize
	\begin{align*}
		Reg_i(T) 
		&\leq \left(\frac{192 N \log T}{\Delta_{\min}^2} + N \log \left(\frac{192 N \log T}{\Delta_{\min}^2}\right) + 3 N^2\right) \cdot \Delta_{i, \max} \\
		&= \gO \left(\frac{N \log T}{\Delta_{\min}^2}\right)
	\end{align*}
}
\end{thm}

The regret bound consists of three key components. First, the cumulative regret from exploration rounds in phase 1. Second, the regret from communication rounds in phase 1. The final term combines two distinct elements: (1) the regret during phase 2, where the YRMH-IGYT mechanism guarantees convergence to the core matching within at most $N$ epochs (with each epoch requiring $\leq N$ rounds to identify a top trading cycle), and (2) the regret contribution from rare concentration failure events.

For convenience, let $\hat{\mU}^{(t)}(i, j), T_{i, j}^{(t)}, UCB_{i, j}^{(t)}, LCB_{i, j}^{(t)}$ be the value of $\hat{\mU}(i, j), T_{i, j}, UCB_{i, j}, LCB_{i, j}$ at the end of round $t$. Define
	{\scriptsize
\begin{align*}
	\gF = \left\{\exists t \in [T], i \in [N], j \in [N]: \abs{\hat{\mU}^{(t)}(i, j) - \mU(i, j)} > \sqrt{\frac{6 \log T}{T_{i, j}^{(t)}}}\right\}
\end{align*}
}
as the bad event that some preference is not estimated well during the horizon.
Since all players would communicate whether they have a precise enough estimation after each sub-phase of phase 1, and determine to enter phase 2 once they find every player include themselves have good-enough estimations, we can conclude that all players would enter phase 2 at the same time. Denote $\ell_{\max}$ as the largest sub-phase number of phase 1. That is to say, players enter in phase 2 at the end of sub-phase $\ell_{\max}$. We then provide the proof of Theorem~\ref{thm:decentralize_upper_bound} as follows.

\begin{proof}[Proof of Theorem~\ref{thm:decentralize_upper_bound}]
	Let $\Delta_{i, \max}$ be the maximum core regret that may be suffered by player $p_i$ in all rounds, we have $\Delta_{i, \max} \leq 1$ by assumption on the utility matrix. The core regret of each player $p_i$ by following Algorithm~\ref{alg:decentralize} satisfies
	{\scriptsize
		\begin{align}
			Reg_i(T) =& \mathbb{E} \left[\sum_{t = 1}^T \left(\mU(i, \mu^*(p_i))\right) - X_i(t)\right] \nonumber \\
			\leq& \mathbb{E} \left[\sum_{t = 1}^T \ind{\mu_t(i) \neq \mu^*(p_i)} \cdot \Delta_{i, \max}\right] \label{eq:unique_core} \\
			\leq& \mathbb{E} \left[\sum_{t = 1}^T \ind{\mu_t(i) \neq \mu^*(p_i)} \middle| \urcorner \gF\right] \cdot \Delta_{i, \max} \nonumber \\
			&+ \sP(\gF) \cdot T \cdot \Delta_{i, \max} \nonumber \\
			\leq& \mathbb{E} \left[\sum_{t = 1}^T \ind{\mu_t(i) \neq \mu^*(p_i)} \middle| \urcorner \gF\right] \cdot \Delta_{i, \max}
			+ 2 N^2 \Delta_{i, \max} \label{eq:bad_event_prob} \\
			\leq& \mathbb{E} \left[\sum_{\ell = 1}^{\ell_{\max}} \left(2^{\ell} + N\right) + N^2 \middle| \urcorner \gF\right] \cdot \Delta_{i, \max}
			+ 2 N^2 \Delta_{i, \max} \label{eq:exploration_bound} \\
			\leq& \left(\frac{192 N \log T}{\Delta_{\min}^2} + N \log \left(\frac{192 N \log T}{\Delta_{\min}^2}\right)\right) \cdot \Delta_{i, \max} \nonumber \\ 
			&+ 3 N^2 \Delta_{i, \max}, \label{eq:exploration_length}
		\end{align}
	}
	where Eq.(\ref{eq:unique_core}) comes from the fact that in a housing market, there is a unique core matching and hence a unique core matching partner $\mu^*(p_i)$ for player $p_i$. Eq.(\ref{eq:bad_event_prob}) holds based on Lemma~\ref{lem:bad_concentration_prob}. Eq.(\ref{eq:exploration_bound}) holds according to Algorithm~\ref{alg:decentralize} and the fact that we need at most $N^2$ rounds for the YRMH-IGYT procedure to reach the equilibrium (Lemma~\ref{lem:yrmh_igyt_step}). Eq.(\ref{eq:exploration_length}) follows from Lemma~\ref{lem:exploration_length}.
\end{proof}


The following technical lemmas underlying our analysis. Proofs are provided in Appendix~\ref{sec:proof_upper_bound_decentralize}. We begin with Lemma~\ref{lem:bad_concentration_prob}, which bounds the probability of erroneous preference estimation.

\begin{lem}[Bad Concentration Event]\label{lem:bad_concentration_prob}
	$\sP \left(\gF\right) \leq \frac{2 N^2}{T}$.
\end{lem}

Lemma~\ref{lem:yrmh_igyt_step} shows that each player $p_i$ finds their core-matched arm $\mu^*(p_i)$ within at most $N^2$ rounds in Phase 2. \looseness=-1

\begin{lem}\label{lem:yrmh_igyt_step}
	Conditional on $\urcorner \gF$, at most $N^2$ rounds are needed in phase 2 before $A_i(t) = \mu^*(p_i)$, and in all of the following rounds, $A_i(t)$ would not be updated and $p_i$ would always be successfully accepted by $\mu^*(p_i)$.
\end{lem}

Lemma~\ref{lem:exploration_length} bounds the duration of the exploration phase.

\begin{lem}\label{lem:exploration_length}
	Condition on $\urcorner \gF$, phase 1 will proceed in at most $\ell_{\max}$ sub-phases where 
	{\footnotesize
	\begin{align}\label{eq:explore_epoch}
		\ell_{\max} = \min \left\{\ell: \sum_{\ell' = 1}^{\ell} 2^{\ell'} \geq \frac{96 N \log T}{\Delta_{\min}^2}\right\},
	\end{align}
}
\end{lem}
which implies that $\sum_{\ell' = 1}^{\ell_{\max}} 2^{\ell'} \leq \frac{192 N \log T}{\Delta_{\min}^2}$ and $\ell_{\max} = \log \left(\frac{192 N \log T}{\Delta_{\min}^2}\right)$ since the sub-phase length grows exponentially. And all players will enter in phase 2 simultaneously at the end of sub-phase $\ell_{\max}$.

\section{Regret Lower Bound}\label{sec:lower_bound}
We establish a regret lower bound by adapting the method of \citet{auer2002nonstochastic} to our multi-player setting, demonstrating the tightness of the regret upper bound of Algorithm~\ref{alg:decentralize}. Let $Reg(T; \boldsymbol{\nu}, \pi)$ be the cumulative expected regret of policy $\pi$ over all players and time $T$, for an instance with arm distributions $\boldsymbol{\nu} = \left\{\nu_{ij}: i, j \in [N]\right\}$. Denote by $\gP$ the set of all probability distributions with support in $[0, 1]$. For any $\boldsymbol{\nu} \in \gP$, define $D_{\inf}(\boldsymbol{\nu}, x, \gP) = \inf_{\boldsymbol{\nu}' \in \gP} \left\{D(\boldsymbol{\nu}, \boldsymbol{\nu}'): \rho(\boldsymbol{\nu}') > x\right\}$, where $\rho: \gP \rightarrow \sR$ maps a distribution to its mean and $D(\cdot, \cdot)$ is the KL divergence.

\begin{defn}[Uniformly Consistent Policies]\label{def:consistent_policy}
	A policy $\pi$ is uniformly consistent \emph{if and only if} for all $\boldsymbol{\nu} \in \gP$, all $\alpha \in (0, 1)$, the regret $\limsup_{T \to \infty} \frac{Reg(T; \boldsymbol{\nu}, \pi)}{T^{\alpha}} = 0$.
\end{defn}
This ensures the policy is not overly tuned to the current instance at the expense of performance in others, ensuring robust performance across different problem instances~\citep{auer2002nonstochastic,lattimore2020bandit}.

\paragraph{Single Top Trading Cycle Bandits} 
Our regret lower bound applies to a subclass of bandits characterized by a single top trading cycle (STTCB). In these instances, each player has a unique top-ranked arm, and there exists a permutation $\sigma$ of players such that $\mu^*(p_{\sigma_i}) = a_{\sigma_{i + 1}}$ for $i = 1, \cdots, N - 1$ and $\mu^*(p_{\sigma_N}) = a_{\sigma_1}$. This structure ensures all players form one top trading cycle in the core matching. For each player $p_i$, define the gap $\Delta_{j}^{(i)} = \mU(i, \mu^*(p_i)) - \mU(i, j)$ and the minimum gap $\Delta_{\min}^{(i)} = \min_{a_j \neq \mu^*(p_i)} \mU(i, \mu^*(p_i)) - \mU(i, j)$, which are non-negative in STTCB instances.

\begin{lem}[Regret Decomposition]\label{lem:regret_decompose}
	For an STTCB instance $\boldsymbol{\nu} = \left\{\nu_{i, j}: i \in [N], j \in [N]\right\}$, and any uniformly consistent policy $\pi$, player $p_i$ for $i \in [N]$ the following holds
	{\scriptsize
		\begin{align*}
			Reg_i(T; \boldsymbol{\nu}, \pi) \geq \max \Biggl\{&\sum_{i' \neq i} 	\Delta_{\min}^{(i)} \mathbb{E}_{\boldsymbol{\nu}, \pi} \left[N_{\mu^*(p_i)}^{(i')}(T)\right],\\
			&\sum_{a_j \neq \mu^*(p_i)} \Delta_{j}^{(i)} \mathbb{E}_{\boldsymbol{\nu}, \pi} \left[N_{j}^{(i)}(T)\right]\Biggr\}.
		\end{align*}
	}
\end{lem}

For an STTCB instance, when the $p_i$'s core-matched partner $\mu^*(p_i)$ is allocated to another player, the optimal outcome is for $p_i$ to receive its second-best arm. While this event occurs infrequently, it provides a sufficient bound for the first term above.

\begin{thm}\label{thm:lower_bound}
	For any player $p_i$, $i \in [N]$, and under any decentralized uniformly consistent algorithm $\pi$, the performance on an STTCB instance $\boldsymbol{\nu}$ satisfies
	{\footnotesize
		\begin{align}\label{eq:lower_bound}
			\begin{split}
				&\liminf_{T \to \infty} \frac{Reg_i(T; \boldsymbol{\nu}, \pi)}{\log T} \\
				\geq \max \Biggl\{&\sum_{i' \neq i} \frac{\Delta_{\min}^{(i)}}{D_{\inf}(\boldsymbol{\nu}_{i', \mu^*(p_i)}, \mU(i', \mu^*(p_{i'})), \gP)}, \\ 
				&\sum_{a_j \neq \mu^*(p_i)} \frac{\Delta_{j}^{(i)}}{D_{\inf}(\boldsymbol{\nu}_{i, j}, \mU(i, \mu^*(p_i)), \gP)}\Biggr\}.
			\end{split}
		\end{align}
	}
\end{thm}

\begin{proof}[Proof Idea]
	For a given STTCB instance $\boldsymbol{\nu}$, we construct a confounding alternative $\boldsymbol{\nu}'$ that differs only in the utility of a single player-arm pair $(p_i, a_j)$, where $a_j$ is the core-matched arm of $p_i$ in $\boldsymbol{\nu}'$ but not in $\boldsymbol{\nu}$. To identify the true instance and avoid linear regret, any algorithm must gather enough evidence to distinguish $\boldsymbol{\nu}$ from $\boldsymbol{\nu}'$, necessitating a number of samples inversely proportional to their KL-divergence. This directly implies a logarithmic regret lower bound.
\end{proof}

The detailed proof of Theorem~\ref{thm:lower_bound} can be found in Appendix~\ref{sec:proof_lower_bound}.
The following corollary demonstrates that the $N \log T / \Delta_{\min}^2$ dependence in Theorem~\ref{thm:decentralize_upper_bound} cannot be improved, thus establishing $\Theta(N \log T / \Delta_{\min}^2)$ as the minimax regret rate for all players.

\begin{coro}\label{coro:lower_bound}
	There exists an STTCB instance with Bernoulli rewards, where the regret of player $p_i$ is lower bounded as $\Omega(\frac{N \log T}{\Delta_{\min}^2})$.
\end{coro}

Reward heterogeneity necessitates $\Omega(\log T / \Delta_{\min}^2)$ explorations for the player with minimum gap $\Delta_{\min}$, during which players with substantial gaps ($\Delta_{\min}^{(i')} = \Omega(1)$) incur significant regret. This core observation underpins our proof.

From Corollary~\ref{coro:lower_bound}, we can see that the regret upper bound (Theorem~\ref{thm:decentralize_upper_bound}) for Algorithm~\ref{alg:decentralize} matches the regret lower bound, showing the significance of our proposed algorithm.

\section{Centralized Anytime Algorithm}\label{sec:centralize}
While Algorithm~\ref{alg:decentralize} achieves a regret upper bound that matches the lower bound in Theorem~\ref{thm:lower_bound} for all players, it requires prior knowledge of the time horizon $T$ to properly construct confidence intervals. Although one could apply the doubling trick~\citep{auer1995gambling} to make the algorithm anytime, this approach leads to non-monotonic expected instantaneous regret. Such behavior could raise concerns among players about the algorithm's consistency during execution.

In this section, we develop an adaptive anytime algorithm that operates through a centralized platform capable of aggregating player preferences and regulating allocations without observing individual reward realizations. The algorithm employs the principle of optimism in the face of uncertainty, where each player maintains upper confidence bounds (UCB) to rank arms and submits these rankings to the platform every round. The platform then computes a core matching by applying the TTC algorithm to the collected preferences, and players subsequently pull their assigned arms. 

We begin by formally defining the UCB estimation method used by individual players and introducing several technical concepts necessary for the analysis. Building on this foundation, we derive a regret upper bound for the centralized approach. A key feature of this framework is that the platform's coordination eliminates potential conflicts between player proposals, allowing us to assume collision-free operation throughout the learning process.

\begin{algorithm}[tp]
	{\small
	\caption{Centralized Anytime UCB}\label{alg:centralize}
	\begin{algorithmic}[1]
		\Require player set $\gP$ and the corresponding arm set $\gA$.
		\For{$t = 1, \cdots$}
		\State The \emph{platform} receives ranking $\hat{r}_{i, t}$ from all players $p_i$.
		\State The \emph{platform} computes a core matching $\mu_t$ using the offline TTC algorithm with the ranking $\hat{r}_{i, t}$.
		\For{$i = 1, \cdots, N$}//Players pull arms simultaneously
		\State $A_i(t) = \mu_t(p_i).$
		\State Update $\hat{\mU}^{(t)}(i, A_i(t))$ and $u^{(t)}(i, A_i(t))$ according to Eq.(\ref{eq:empirical_mean}) and Eq.(\ref{eq:ucb}).
		\State Compute the current ranking $\hat{r}_{i, t + 1}$ according to $u^{(t)}(i, \cdot)$.
		\EndFor
		\EndFor
	\end{algorithmic}
}
\end{algorithm}

At iteration $t$, when a player $p_i$ gets matched to arm $A_i(t)$, it would update the estimated preference value $\hat{\mU}(i, A_i(t))$ and the observed time $T_{i, A_i(t)}$ for arm $A_i(t)$ according to Eq.(\ref{eq:empirical_mean}) in Section~\ref{sec:decentralize}.

Then the upper confidence bound, which is called the \emph{index} for each $(i, j)$-th entry of the utility matrix $\mU$ is computed as
{\footnotesize
\begin{align}\label{eq:ucb}
	u^{(t)}(i, j) = \begin{cases}
		\infty  & \text{if } T_{i, j}(t) = 0,\\
		\hat{\mU}^{(t)}(i, j) + \sqrt{\frac{3  \log t}{2 T_{i, j}(t - 1)}}, & \text{ otherwise.}
	\end{cases}
\end{align}
}

Each player $p_i$ uses the index to compute a preference ranking $\hat{r}_{i, t + 1}$, where arms are ordered by their UCBs in a decreasing order (e.g., $\argmax_ju^{(t)}(i, j)$ is ranked first).

Let $T_{\mu}(t)$ denote the number of times a matching $\mu$ is played by time $t$. A matching is \emph{achievable} at time $t$ if it is core-stable according to the current estimated rankings $\left\{\hat{r}_{i, t}\right\}_{i \in [N]}$. A matching is \emph{truly core-stable} if it is core-stable under the true utility matrix $\mU$. 
For player $p_i$ and arm $a_j$, define $M_{i, j}$ as the set of achievable (but not truly core-stable) matchings where $p_i$ is matched to $a_j$. The \emph{achievable preference gap} is defined as $\bar{\Delta}_{i, j} = \mU(i, \mu^*(p_i)) - \mU(i, j)$.

Since the truly core-stable matching $\mu^*$ incurs zero regret, we bound the regret of player $p_i$ as follows:
{\scriptsize
\begin{align}\label{eq:achievable_gap_bound}
	\begin{split}
		Reg_i(T) &\leq \sum_{j: \bar{\Delta}_{i, j} > 0} \bar{\Delta}_{i, j} \left(\sum_{\mu \in M_{i, j}} \mathbb{E} T_{\mu}(T)\right) \\
		&\leq \max_{j} \bar{\Delta}_{i, j} \left(\sum_{\mu \in M} \mathbb{E} T_{\mu}(T) \right),
	\end{split}
\end{align}}
where $M$ is the set of all achievable but non-truly core-stable matchings, i.e., $M = \cup_{(i, j) \in [N] \times [N]} M_{i, j}$.
If a matching $\mu$ is not truly core-stable under the true utility matrix $\mU$, a \emph{blocking coalition} must exist -- a subset of players who can reassign their endowed arms to strictly improve every member's utility. 
For any non-truly-core-stable matching $\mu \in M$, there exists a blocking coalition $\gB_{\mu}$ such that players outside $\gB_{\mu}$ keep their assignments under $\mu^*$. Within $\gB_{\mu}$, at least one player $p_m$ must be part of a \emph{blocking triplet} $(p_m, a_n, a_{n'})$ where under $\mu$, $p_m$ is assigned $a_n$, $\mU(m, n) < \mU(m, n')$ in truth yet the UCB index satisfies $u(m, n) > u(m, n')$. 
If no such triplet existed, then for every matched pair $(p_m, a_n)$ and every $a_{n'}$ with $\mU(m, n) < \mU(m, n')$, we would have $u(m, n) < u(m, n')$. Since $\mu$ is non-core-stable under $\mU$, there must exist a reallocation within $\gB_{\mu}$ that strictly improves every member with respect to $\mU$ -- and, given the UCB ordering, also with respect to the UCB index. This would contradict the UCB-core-stability of the implemented matching.
Let $Q_{m, n}$ be the collection of all such triplets, whenever a non-truly-core-stable $\mu \in M$ is implemented, at least one player-arm pair $(p_m, a_n) \in \gB := \cup_{\mu \in M} \gB_{\mu}$ must realize a triplet from $Q_{m, n}$. This structure yields the following regret bound.

\begin{thm}\label{thm:centralize_upper_bound}
	Following Algorithm~\ref{alg:centralize}, the core regret of player $p_i$ up to time $T$ satisfies 
	{\scriptsize
		\begin{align*}
			Reg_i(T) &\leq \max_{j} \bar{\Delta}_{i, j}\left[\sum_{\substack{(p_m, a_n) \\ \in \gB}} \sum_{\substack{(p_m, a_n, a_{n'}) \\ \in Q_{m, n}}} \left(5 + \frac{6 \log T}{\Delta_{m, n', n}^2}\right)\right] \\
			&\leq \max_{j} \bar{\Delta}_{i, j} \left(5 N^3 + 12 \frac{N^2 \log T}{\Delta_{\min}^2}\right) = \gO \left(\frac{N^2 \log T}{\Delta_{\min}^2}\right).
		\end{align*}
	}
\end{thm}

Theorem~\ref{thm:centralize_upper_bound} offers a centralized problem-dependent $\gO\left(\frac{N^2 \log T}{\Delta_{\min}^2}\right)$ upper bound guarantees on the core regret of each player $p_i$, and the proposed UCB-type algorithm is anytime and adaptive in the sense that the players do not need to know the time horizon $T$ and the minimum preference gap $\Delta_{\min}$ in advance.\looseness=-1

\begin{proof}
	Let $T_{m, n, n'}(T)$ be the number of times player $p_m$ pulls arm $a_n$ while $(p_m, a_n, a_{n'})$ forms a blocking triplet. Since every time a matching $\mu \in M$ is implemented, at least one such blocking triplet $(p_m, a_n, a_{n'}) \in Q_{m, n}$ occurs for some $(p_m, a_n) \in \gB$, we have
	{\footnotesize
		\begin{align}\label{eq:relation_blocking}
			\sum_{\mu \in M} T_{\mu}(T) \leq \sum_{\substack{(p_m, a_n) \\ \in \gB}} \sum_{\substack{(p_m, a_n, a_{n'}) \\ \in Q_{m, n}}} T_{m, n, n'}(T).
		\end{align}
	}
	From the analysis above, we know that when $(p_m, a_n, a_{n'})$ is blocking, we have the UCB index $u(m, n) > u(m, n')$ while $\mU(m, n) < \mU(m, n')$ according to the ground-truth utility matrix.
	Standard analysis for the single player UCB~\citep{bubeck2012regret} shows that
	{\footnotesize 
		\begin{align}\label{eq:blocking_triplet_time}
			\mathbb{E} T_{m,n, n'}(T) \leq 5 + \frac{6 \log T}{\Delta^2_{m, n', n}}.
		\end{align}
	}
	Combining Eq.(\ref{eq:achievable_gap_bound}), (\ref{eq:relation_blocking}) and (\ref{eq:blocking_triplet_time}) we get the first inequality on the regret upper bound.
	
	When we consider the triplet set composed of all possible $\mU(m, n) < \mU(m, n')$, we have $\abs{Q_{m, n}} \leq N$ and 
	{\scriptsize
		\begin{align*}
			\sum_{n': \mU(m, n) < \mU(m, n')} \frac{1}{\Delta_{m, n', n}^2} \leq \sum_{n' = 1}^N \frac{1}{(n')^2 \Delta_{\min}^2} \leq \frac{2}{\Delta_{\min}^2},
		\end{align*}
	}
	and hence
	{\scriptsize
		\begin{align*}
			&\max_j \bar{\Delta}_{i, j} \left[\sum_{(p_m, a_n) \in \gB} \sum_{(p_m, a_n, a_{n'}) \in Q_{m, n}} \left(5 + \frac{6 \log T}{\Delta_{m, n', n}^2}\right)\right] \\
			\leq &\max_j \bar{\Delta}_{i, j} \left[\sum_{(p_m, a_n) \in \gB} \left(5 N + \frac{12 \log T}{\Delta_{\min}^2}\right)\right].
		\end{align*}
	}
	
	As there are at most $N^2$ possible player-arm pairs in the blocking coalition, the second inequality is concluded.
\end{proof}

\section{Conclusion and Discussion}\label{sec:conclusion}
This paper introduces a novel online learning framework for housing markets with uncertain preferences, unifying two key objectives: core-stability and sample efficiency. We formalize this through core regret as our performance metric and propose two algorithms that bridge multi-armed bandit techniques with housing market mechanisms. These algorithms adapt to both centralized and decentralized settings, with guarantees for fixed-horizon and anytime environments. A matching regret lower bound proves the order-optimality of the decentralized algorithm.

There are many additional questions that can be studied in this model.

\paragraph{Housing Markets with Existing Tenants}
This paper focuses on housing markets where each player initially possesses one endowed arm, maintaining an equal number of players and arms. However, real-world housing markets typically involve more complex scenarios: existing tenants with endowed arms coexist with new applicants lacking initial allocations, while vacant houses (free arms) may also be available~\citep{abdulkadirouglu1999house}. In these generalized settings, the core loses its uniqueness, raising two fundamental challenges: first, identifying an appropriate tractable solution concept to serve as a benchmark, and second, determining whether efficient learning algorithms can achieve sublinear regret in this more realistic but complicated environment. \looseness=-1

\paragraph{Housing Markets with Indifference}
While our current analysis assumes strict player preferences over arms, real-world housing markets often involve indifference between options. Our learning algorithm achieves a regret bound of $\Theta(N \log T / \Delta_{\min}^2)$, which becomes vacuous when $\Delta_{\min} = o(1 / \sqrt{T})$ -- precisely when indifference between arms creates vanishing preference gaps. This limitation highlights the need for new algorithmic approaches that can handle indifferences in housing market allocations, presenting an important direction for future research. 

\paragraph{Incentive Compatibility in the Learning Setting}
While the top trading cycle algorithm is strategy-proof under known, deterministic preferences, its incentive compatibility remains unclear in learning settings where preferences must be discovered dynamically. In particular, the decentralized nature of these interactions may create opportunities for strategic manipulation through learning behavior. Understanding whether and how players can exploit the learning process to achieve better outcomes presents a significant open question for future research.

\bibliography{reference}

\appendix
\onecolumn

\section{Proof of Regret Upper Bound of the Decentralized Algorithm}\label{sec:proof_upper_bound_decentralize}
\subsection{Proof of Lemma~\ref{lem:bad_concentration_prob}}\label{subsec:proof_bad_concentration_prob}
\begin{proof}
	\begin{align*}
		\sP(\gF) &= \sP \left(\exists 1 \leq t \leq T, i \in [N], j \in [N]: \abs{\hat{\mU}^{(t)}(i, j) - \mU(i, j)} > \sqrt{\frac{6 \log T}{T_{i, j}^{(t)}}}\right) \\
		&\leq \sum_{t = 1}^T \sum_{i \in [N]} \sum_{j \in [N]} \sP \left(\abs{\hat{\mU}^{(t)}(i, j) - \mU(i, j)} > \sqrt{\frac{6 \log T}{T_{i, j}^{(t)}}}\right) \\
		&\leq \sum_{t = 1}^T \sum_{i \in [N]} \sum_{j \in [N]} \sum_{s = 1}^t \sP \left(T_{i, j}^{(t)} = s, \abs{\hat{\mU}^{(t)}(i, j) - \mU(i, j)} > \sqrt{\frac{6 \log T}{s}}\right) \\
		&\leq \sum_{t = 1}^T \sum_{i \in [N]} \sum_{j \in [N]} t \cdot 2 \exp \left(-3 \log T\right) \\
		&\leq \frac{2 N^2}{T},
	\end{align*}
	where the second last inequality results from Lemma~\ref{lem:concentration}.
\end{proof}

\subsection{Proof of Lemma~\ref{lem:yrmh_igyt_step}}\label{subsec:proof_yrmh_igyt_step}
\begin{proof}
	According to Lemma~\ref{lem:ucb_lcb} and Algorithm~\ref{alg:decentralize}, when player $p_i$ enters in phase 2 with ranking $\sigma_i$, we have
	\begin{align*}
		\mU(i, \sigma_{i, j}) > \mU(i, \sigma_{i, j + 1}), \quad \forall j \in [N].
	\end{align*}
	That is to say, player $p_i$ successfully recovers the real preference ranking when she enters in phase 2. Further, according to Lemma~\ref{lem:exploration_length}, all players would enter in phase 2 simultaneously. Based on Lemma~\ref{lem:offline_yrmh_igyt}, we know that when given the true preference profile, YRMH-IGYT would stop in at most $N^2$ steps, and once the core matching is reached, players would focus on their core matching partner in all the remaining rounds. Since there is a unique core matching and no two players would share their core matching partners, $p_i$ could always get matched to $\mu^*(p_i)$.
\end{proof}

\begin{lem}\label{lem:ucb_lcb}
	Condition on $\urcorner \gF$, $\text{UCB}_{i, j}^{(t)} < \text{LCB}_{i, j'}^{(t)}$ implies $\mU(i, j) < \mU(i, j')$.
\end{lem}

\begin{proof}
	According to the definition of LCB and UCB, we have that conditional on $\urcorner \gF$,
	\begin{align*}
		\text{LCB}_{i, j}^{(t)} &= \hat{mU}_{i, j}^{(t)} - \sqrt{\frac{6 \log T}{T_{i, j}^{(t)}}} \\
		&\leq \mU(i, j) \\
		&\leq \hat{\mU}^{(t)}(i, j) + \sqrt{\frac{6 \log T}{T_{i, j}^{(t)}}} \\
		&= \text{UCB}_{i, j}^{(t)}.
	\end{align*}
	Therefore, if $\text{UCB}_{i, j}^{(t)} < \text{LCB}_{i, j'}^{(t)}$, we have that 
	\begin{align*}
		\mU(i, j) \leq \text{UCB}_{i, j}^{(t)} \leq \text{LCB}_{i, j}^{(t)} \leq \mU(i, j').
	\end{align*}
\end{proof}

\begin{lem}[YRMH-IGYT Procedure]\label{lem:offline_yrmh_igyt}
	Given the true preferences of all players, at most $N$ epochs are needed to find the core matching, while in each epoch, at most $N$ rounds are needed to reach a top trading cycle. Therefore, at most $N^2$ steps are needed for the YRMH-IGYT procedure to detect the core matching.
\end{lem}

\begin{proof}
	In the YRMH-IGYT procedure, each epoch proceeds as follows:
	\begin{enumerate}
		\item Leader Selection: At the start of each epoch, the player with the smallest-index endowed arm among the remaining available arms is chosen as the leader to initiate the epoch.
		\item Proposal Dynamics: The leader begins by proposing to their most preferred available arm. For subsequent rounds, any player who receives a proposal in round $t$ will, in round $t + 1$, propose to their own most favored remaining arm. If a player receives no proposal, they remain passive and do not make any offers. Critically, exactly one player proposes to one arm in every round, ensuring a structured progression.
		\item Termination Condition: An epoch terminates when a player who previously made a proposal receives a new proposal from another player. This player must belong to a top trading cycle (TTC), as demonstrated in Lemma~\ref{lem:cycle_detection}.
	\end{enumerate}
	
	We could observe that each epoch identifies exactly one TTC, as the procedure ensures players always propose to their top remaining choices. Then in the worst case, in a single epoch, the request path may involve all $N$ players, requiring up to $N$ steps. On the other hand, since each TTC has length $\geq 1$, at most $N$ epochs are needed, leading to an overall worst-case runtime of $N^2$ steps.
\end{proof}

\begin{lem}[Cycle Detection]\label{lem:cycle_detection}
	In any epoch during the YRMH-IGYT procedure, the first repeated player in the request path belongs to a top trading cycle.
\end{lem}

\begin{proof}
	To formalize, consider the request path generated during the procedure:
	\begin{itemize}
		\item If a player $p_i$ proposes to an arm owned by $p_j$, then $p_j$ (upon receiving the proposal) will propose to their favorite remaining arm, say one owned by $p_k$. This creates a chain $p_i \rightarrow p_j \rightarrow p_k \rightarrow \cdots$.
		\item A cycle emerges when this chain loops back to a player already in the sequence (e.g. $p_k \rightarrow p_i$).
	\end{itemize}
	In a housing market, the cycle is the top trading cycle, and it is detected when a player who previously proposed receives a new proposal, closing the loop.
\end{proof}

\subsection{Proof of Lemma~\ref{lem:exploration_length}}\label{subsec:proof_exploration_length}
\begin{proof}
	Since players propose to arms based on the indices of their own endowed arm in a round-robin way, no collision occurs and all players can be successfully accepted at each round in the exploration stage of phase 1. Therefore, at the end of the sub-phase $\ell_{\max}$ defined in Eq.(\ref{eq:explore_epoch}), it holds that $T_{i, j} \geq 96 \log T / \Delta_{\min}^2$ for any $i, j \in [N]$.
		
	By Lemma~\ref{lem:explore_time}, once $T_{i, j} \geq 96 \log T / \Delta_{\min}^2$ for all arms $a_j$, player $p_i$ can recover the true preference ranking under the good concentration event. That is, $p_i$ obtains a permutation $\sigma_i$ satisfying $\text{LCB}_{i, \sigma{i, j}} > \text{UCB}_{i, \sigma{i, j+1}}$ for all $j \in [N]$. Consequently, during the communication stage of sub-phase $\ell_{\max}$, every player will propose to arm $a_j$ in the $j$-th communication round. This ensures that in the $i$-th communication round, player $p_i$ receives $N$ proposals, confirming that all players have successfully learned the true preference profile. Thus, all players synchronously transition to Phase 2 at the end of sub-phase $\ell_{\max}$.
\end{proof}

\begin{lem}\label{lem:explore_time}
	In round $t$, let $T_{i}^{(t)} = \min_{j \in [K]} T_{i, j}^{(t)}$. Conditional on $\urcorner \gF$, if $T_i^{(t)} > 96 \log T / \Delta_{\min}^2$, we have $LCB_{i, r_{i, k}} > UCB_{i, r_{i, k + 1}}^{(t)}$ for any $k \in [N]$.
\end{lem}

\begin{proof}
	We prove it by contradiction, suppose that there exists $k \in [N]$ such that $LCB_{i, r_{i, k}} \leq UCB_{i, r_{i, k + 1}}^{(t)}$. Without loss of generality, denote $j$ as the arm on the LHS and $j'$ as the arm on the RHS.
	Conditional on $\urcorner \gF$ and by the definition of LCB and UCB, we have that 
	\begin{align*}
		&\mU(i, j) - 2 \sqrt{\frac{6 \log T}{T_{i}^{(t)}}} \\
		\leq& \text{LCB}_{i, j}^{(t)}\\
		\leq& \text{UCB}_{i, j'}^{(t)} \\
		\leq& \mU(i, j') + 2 \sqrt{\frac{6 \log T}{T_i^{(t)}}}.
	\end{align*}
	Therefore, $\Delta_{i, j, j'} = \mU(i, j) - \mU(i, j') \leq 4 \sqrt{6 \log T / T_i^{(t)}}$, which implies that $T_i^{(t)} \leq 96 \log T / \Delta_{i, j, j'}^2 \leq 96 \log T / \Delta_{\min}^2$, which is a contradiction.
\end{proof}

\section{Proof of Regret Lower Bound}\label{sec:proof_lower_bound}
We will use the following notations throughout the proof of the lower bound.

\begin{itemize}
	\item For any time $t$, the number of times arm $j \in [N]$ is played by player $p_i$ is $N_{j}^{(i)}(t)$. 
	
	\item Distribution of player $i \in [N]$ and arm $j \in [N]$ is given by $\boldsymbol{\nu}_{i,j}$, which has mean $\mU(i, j)$. And we assume that for all $i, j \in [N]$, $\mU(i, j) \in [0, 1]$.
	
	\item The core matching partner of any player $i \in [N]$ is given by $\mu^*(p_i)$. 
	
	\item For any player $i \in [N]$, arm $j \in [N]$, $\Delta_{j}^{(i)} := \mU(i, \mu^*(p_i)) - \mU(i, j)$, the arm gap. This can be negative.
\end{itemize}

\subsection{Divergence Decomposition}\label{subsec:divergence_decompose}
The proof of the divergence decomposition lemma builds upon the framework presented in Chapter 16 of \citet{lattimore2020bandit}, extending it to accommodate multiple players. We first introduce some notations as follows.

\paragraph{Canonical multi-player with endowed arm bandit model:}
We now define the ($N$-player, $N$-arm, $T$-horizon) bandit models, where each player is associated with an arm. The canonical bandit model lies in a measurable space $\left\{\Omega, \gF\right\}$. Let $A_i(t)$ denote the arm chosen by player $p_i$ at time $t$. We denote the rejection by the symbol $\emptyset$. Therefore, $A_i(t) \in [N]$, and $X_i(t) \in \sR \cup \left\{\emptyset\right\}$ for all $i \in [N]$ and $t \in [T]$. Also, $A_i(t)$ and $X_i(t)$ for all $i \in [N]$ and $t \in [T]$ are measurable with respect to $\gF$. Let $H(t) = \left(A_i(t'), X_i(t'): \forall i \in [N], \forall t' \leq t\right)$ be the random variable representing the history of actions taken and rewards seen up to and including time $t$. We have $H(t) \in \gH(t) \equiv \left([N]^N \times \left(\sR \cup \left\{\emptyset\right\}\right)^N\right)^t$. We may set $\Omega \equiv \gH(T)$ and the sigma algebra generated by the history as $\gF \equiv \sigma \left(H(T)\right)$.

\paragraph{Environment:}
The bandit environment is specified by $\boldsymbol{\nu} = \left(\boldsymbol{\nu}_{i, j}: \forall i \in [N], \forall j \in [N]\right)$, where $\boldsymbol{\nu}_{i, j}$ is the distribution of rewards obtained when arm $a_j$ is matched to player $p_i$ in this environment.

\paragraph{Policy:}
A policy is a sequence of distribution of possible request to the arms from the players (which can assimilate any coordination among the players) conditioned on the past events. More formally, the policy $\boldsymbol{\pi} = \left\{\boldsymbol{\pi}_t(\cdot): t \in [T]\right\}$ where $\boldsymbol{\pi} \equiv \left\{\pi_t(i, j | \cdot): \forall i \in [N], \forall j \in [N]\right\}$ is the function that maps the history up to time $t - 1$ to the action $A_i(t), \forall i \in [N]$. Further, $\pi_t(i, j | \cdot): \gH(t - 1) \rightarrow [0, 1]$ denotes the probability, as a function of history $H(t - 1)$ of player $p_i$ playing arm $a_j$.

\paragraph{Probability Measure:}
Each environment $\boldsymbol{\nu}$ and policy $\boldsymbol{\pi}$ jointly induces a probability distribution over the measurable space $\left\{\Omega, \gF\right\}$ denoted by $\sP_{\boldsymbol{\nu}, \boldsymbol{\pi}}$. Let $\mathbb{E}_{\boldsymbol{\nu}, \boldsymbol{\pi}}$ denote the expectation induced. The density of a particular history up to time $T$, under an environment $\boldsymbol{\nu}$ and a policy $\boldsymbol{\pi}$, can be defined as 
\begin{align*}
	&d \sP_{\boldsymbol{\nu}, \boldsymbol{\pi}} \left(\boldsymbol{A}(t), \boldsymbol{X}(t): t \in [T]\right) \\
	=& \prod_{t = 1}^T \pi_t(\boldsymbol{A}(t) | h(t - 1)) p_{\nu} \left(\boldsymbol{X}(t) | \boldsymbol{A}(t)\right) d \lambda \left(\boldsymbol{X}(t); \boldsymbol{\nu}\right) d \varrho \left(\boldsymbol{A}(t)\right).
\end{align*}
Here, $\lambda(\boldsymbol{X}; \boldsymbol{\nu})$ = $\prod_{i = 1}^N \lambda_{i} \left(X_i\right)$ is the dominating measure over the rewards with $\lambda_i(X_i) = \delta_{\emptyset} + \sum_{j} \boldsymbol{\nu}_{i,j}$\footnote{Here $\delta_{\emptyset}$ is the dirac measure on $\emptyset$ denoting the rejection event.}. Also, $\varrho(\boldsymbol{A})$ is the counting measure on the collective action of the players.

\begin{lem}[Divergence Decomposition]\label{lem:divergence_decompose}
	For two bandit instances $\boldsymbol{\nu} = \left\{\boldsymbol{\nu}_{i, j}: i \in [N], j \in [N]\right\}$, and $\boldsymbol{\nu}' = \left\{\boldsymbol{\nu}'_{i, j}: i \in [N], j \in [N]\right\}$, and any admissible policy $\boldsymbol{\pi}$ the following divergence decomposition is true
	\begin{align*}
		D(\sP_{\boldsymbol{\nu}, \boldsymbol{\pi}}, \sP_{\boldsymbol{\nu}', \boldsymbol{\pi}}) = \sum_{i = 1}^N \sum_{j = 1}^N \mathbb{E}_{\boldsymbol{\nu}, \boldsymbol{\pi}} \left[N_j^{(i)}(T)\right] D \left(\boldsymbol{\nu}_{i, j}, \boldsymbol{\nu}'_{i, j}\right).
	\end{align*}
\end{lem}

\begin{proof}
	The divergence between two measure, which correspond to two different environments under a policy $\boldsymbol{\pi}$, $\sP_{\boldsymbol{\nu}, \boldsymbol{\pi}}$ and $\sP_{\boldsymbol{\nu}', \boldsymbol{\pi}}$ can be expressed as 
	\begin{align}
		&D \left(\sP_{\boldsymbol{\nu}, \boldsymbol{\pi}}, \sP_{\boldsymbol{\nu}', \boldsymbol{\pi}}\right) \nonumber \\
		=& \mathbb{E}_{\sP_{\boldsymbol{\nu}, \boldsymbol{\pi}}} \left[\sum_{t = 1}^T \log \left(\frac{d \sP_{\sP_{\boldsymbol{\nu}, \boldsymbol{\pi}}}}{d \sP_{\sP_{\boldsymbol{\nu}', \boldsymbol{\pi}}}}\right)\right] \nonumber \\
		=& \mathbb{E}_{\sP_{\boldsymbol{\nu}, \boldsymbol{\pi}}} \left[\log \left(\frac{p_{\boldsymbol{\nu}} (\boldsymbol{X}(t) | \boldsymbol{A}(t))}{p_{\boldsymbol{\nu}'} (\boldsymbol{X}(t) | \boldsymbol{A}(t))}\right)\right] \label{eq:density_cancel} \\
		=& \mathbb{E}_{\sP_{\boldsymbol{\nu}, \boldsymbol{\pi}}} \left[ \sum_{t = 1}^T \log \left(\frac{\prod_{i: X_i(t) \neq \emptyset} p_{\boldsymbol{\nu}} (X_i(t) | \boldsymbol{A}(t))}{\prod_{i: X_i(t)} p_{\boldsymbol{\nu}'} (X_i(t) | \boldsymbol{A}(t))}\right)\right] \label{eq:player_decompose} \\
		=& \mathbb{E}_{\sP_{\boldsymbol{\nu}, \boldsymbol{\pi}}} \left[\sum_{t= 1}^T \sum_{i: X_i(t) \neq \emptyset} \mathbb{E}_{\boldsymbol{\nu}} \left[\log \left(\frac{p_{\boldsymbol{\nu}} (X_i(t) | \boldsymbol{A}(t))}{p_{\boldsymbol{\nu}'} (X_i(t) | \boldsymbol{A}(t))}\right) \middle| \boldsymbol{A}(t)\right]\right] \nonumber \\
		=& \mathbb{E}_{\sP_{\boldsymbol{\nu}, \boldsymbol{\pi}}} \left[\sum_{t = 1}^T \sum_{i: X_i(t) \neq \emptyset} D \left(\boldsymbol{\nu}_{i, A_i(t)}, \boldsymbol{\nu}'_{i, A_i(t)}\right)\right] \label{eq:divergence_def} \\
		=& \sum_{i = 1}^N \sum_{j = 1}^N \mathbb{E}_{\sP_{\boldsymbol{\nu}, \boldsymbol{\pi}}} \left[\sum_{t = 1}^T \ind{A_i(t) = j, X_i(t) \neq \emptyset} D \left(\boldsymbol{\nu}_{i, j}, \boldsymbol{\nu}'_{i, j}\right)\right] \nonumber \\
		=& \sum_{i = 1}^N \sum_{j = 1}^N \mathbb{E}_{\sP_{\boldsymbol{\nu}, \boldsymbol{\pi}}} \left[N_{j}^{(i)}(T)\right] D \left(\boldsymbol{\nu}_{i, j}, \boldsymbol{\nu}'_{i, j}\right), \label{eq:pulling_time}
	\end{align}
	where Eq.(\ref{eq:density_cancel}) comes from the fact that the density of the policy cancels out for the two different environments. Eq.(\ref{eq:player_decompose}) holds because if for some set of actions $\boldsymbol{A}(t)$ player $p_i$ observes $X_i(t) = \emptyset$, it indicates that player $p_i$ is rejected on that round, which is independent of the environment. In particular, we have $p_{\boldsymbol{\nu}} \left(X_i(t) \middle| \boldsymbol{A}(t)\right) = p_{\boldsymbol{\nu}'} \left(X_i(t) \middle| \boldsymbol{A}(t)\right)$ if $X_i(t) = \emptyset$ for any $\boldsymbol{A}(t)$. Eq.(\ref{eq:divergence_def}) comes from the definition of divergence. Eq.(\ref{eq:pulling_time}) uses the definition of $N_j^{(i)}(t)$ the total number of times player $p_i$ successfully plays arm $a_j$ up to time $T$. 
\end{proof}

\subsection{Proof of Regret Decomposition (Lemma~\ref{lem:regret_decompose})}\label{subsec:proof_regret_decompose}
\begin{proof}
	We fix any player $p_i$, $i \in [N]$ for the rest of the proof. Denote $C_i(T)$ as the number of times that player $p_i$ experiences a collision up to time $T$. We have the expected regret for the player $p_i$, under a policy $\boldsymbol{\pi}$ and any bandit instance $\boldsymbol{\nu}$ as 
	\begin{align*}
		Reg_i(T; \boldsymbol{\nu}, \boldsymbol{\pi})
		 = \sum_{j = 1}^N \Delta_{j}^{(i)} \mathbb{E}_{\boldsymbol{\nu}, \boldsymbol{\pi}} \left[N_j^{(i)}(T)\right] + \sum_{j = 1}^N \mU(i, \mu^*(p_i)) \mathbb{E}_{\boldsymbol{\nu}, \boldsymbol{\pi}} \left[C_i(T)\right],
	\end{align*}
	where the first term comes from successfully pulling an arm that is not the core matching partner for player $p_i$ and the second term reflects that for each collision the player $p_i$ obtains 0 reward and hence $\mU(i, \mu^*(p_i))$ regret in expectation. A trivial regret lower bound is
	\begin{align*}
		Reg_i(T; \boldsymbol{\nu}, \boldsymbol{\pi}) &\geq \sum_{j = 1}^N \Delta_j^{(i)} \mathbb{E}_{\boldsymbol{\nu}, \boldsymbol{\pi}} \left[N_j^{(i)}(T)\right] \\
		&= \sum_{a_j \neq \mu^*(p_i)} \Delta_j^{(i)} \mathbb{E}_{\boldsymbol{\nu}, \boldsymbol{\pi}} \left[N_j^{(i)}(T)\right].
	\end{align*}
	
	For an STTCB instance, the core matching partner of a player is her optimal arm, and if two or more players propose to the same arm in one round, it would result in a collision and all of these players would get a reward of 0. Therefore, if any of the other $N - 1$ players propose to pull $\mu^*(p_i)$, player $p_i$ should either move to a sub-optimal arm or it is blocked. In the best possible scenario, the player $p_i$ successfully plays its second best arm, in each of these instances. We have $\Delta_{\min}^{(i)} \leq \mU(i, \mu^*(p_i))$ for non-negative rewards, denote $\bar{N}_{\mu^*(p_i)^{(i')}}$ as the number of times player $p_{i'}$ proposes to arm $\mu^*(p_i)$, the regret from the events when any of the other $N - 1$ players propose to pull arm $\mu^*(p_i)$ is lower bounded by
	\begin{align*}
		Reg_i(T; \boldsymbol{\nu}, \boldsymbol{\pi}) &\geq \sum_{i' \neq i} \Delta_{\min}^{(i)} \mathbb{E}_{\boldsymbol{\nu}, \boldsymbol{\pi}} \left[\bar{N}_{\mu^*(p_i)}^{(i')}(T)\right] \\
		&\geq \sum_{i' \neq i} \Delta_{\min}^{(i)} \mathbb{E}_{\boldsymbol{\nu}, \boldsymbol{\pi}} \left[N_{\mu^*(p_i)}^{(i')}(T)\right].
	\end{align*}
	The last inequality comes from the fact that player $p_{i'}$ could successfully pull arm $\mu^*(p_i)$ when only one player proposes to this arm.
	Therefore, the combined regret bound is given as
	\begin{align*}
		Reg_i(T; \boldsymbol{\nu}, \pi) \geq \max \Biggl\{\sum_{i' \neq i} 	\Delta_{\min}^{(i)} \mathbb{E}_{\boldsymbol{\nu}, \pi} \left[N_{\mu^*(p_i)}^{(i')}(T)\right],
		\sum_{a_j \neq \mu^*(p_i)} \Delta_{j}^{(i)} \mathbb{E}_{\boldsymbol{\nu}, \pi} \left[N_{j}^{(i)}(T)\right]\Biggr\}.
	\end{align*}
\end{proof}

\subsection{Proof of Regret Lower Bound (Theorem~\ref{thm:lower_bound})}\label{subsec:proof_lower_bound}
\begin{proof}
	We consider any instance in the class of STTCB $\boldsymbol{\nu}$, universally consistent policy $\boldsymbol{\pi}$, player $p_i$ for $i \in [N]$, and arm $a_j$ for $j \in [N]$. Let us consider the instance $\boldsymbol{\nu}'$ (which is specific to the $i$ and $j$ pair) where $\boldsymbol{\nu}'_{i', j'} = \boldsymbol{\nu}_{i', j'}$ for all $i' \neq i, j' \neq j$, and $\boldsymbol{\nu}'_{i, j}$ such that $D \left(\boldsymbol{\nu}_{i, j}, \boldsymbol{\nu}'_{i, j}\right) \leq D_{\inf} \left(\boldsymbol{\nu}_{i, j}, \mU(i, \mu^*(p_i)), \gP\right) + \epsilon$ for some $\epsilon > 0$ and $\mU'(i, j) \equiv \rho(\boldsymbol{\nu}'_{i, j}) > \mU(i, \mu^*(p_i))$ where $\rho$ is the operator mapping a distribution to its mean. Note, for $\mU(i, \mu^*(p_i)) < 1$ and $\Delta_{j}^{(i)} > 0$, which holds by assumption, the distribution $\boldsymbol{\nu}'_{i, j}$ exists by definition of $D_{\inf}(\cdot)$. In short, for the $i$-th player we make the $j$-th arm optimal. The optimal arm for player $i$ in the instance $\boldsymbol{\nu}'$ is the arm $j$.
	
	For any event $A$ (and its complement $A^c$), by Lemma~\ref{lem:BH_inequality} in Appendix~\ref{sec:tech_lemma}, we have 
	\begin{align}\label{eq:BH_inequality}
		D\left(\sP_{\boldsymbol{\nu}, \boldsymbol{\pi}}, \sP_{\boldsymbol{\nu}', \boldsymbol{\pi}}\right) \geq \log \left(\frac{1}{2 \left(\sP_{\boldsymbol{\nu}, \boldsymbol{\pi}}(A) + \sP_{\boldsymbol{\nu}', \boldsymbol{\pi}}(A^c)\right)}\right).
	\end{align}
	
	Let us now consider the event $A = \left\{N_j^{(i)}(T) \geq \frac{T}{2}\right\}$. Then, due to the regret decomposition (Lemma~\ref{lem:regret_decompose}), we have the regrets:
	\begin{enumerate}
		\item In instance $\boldsymbol{\nu}$ as $Reg(T; \boldsymbol{\nu}, \boldsymbol{\pi}) \geq Reg_{i}(T; \boldsymbol{\nu}, \boldsymbol{\pi}) \geq \Delta_{j}^{(i)} \frac{T}{2} \sP_{\boldsymbol{\nu}, \boldsymbol{\pi}} \left(\left\{N_j^{(i)}(T) \geq \frac{T}{2}\right\}\right)$.
		
		\item In instance $\boldsymbol{\nu}'$ as $Reg(T; \boldsymbol{\nu}', \boldsymbol{\pi}) \geq Reg_i(T; \boldsymbol{\nu}', \boldsymbol{\pi}) \geq \left(\mU'(i, j) - \mU(i, \mu^*(p_i))\right) \frac{T}{2} \sP_{\boldsymbol{\nu}, \boldsymbol{\pi}} \left(\left\{N_j^{(i)}(T) < \frac{T}{2}\right\}\right)$.
	\end{enumerate}
	
	As the only change in reward distribution happens in player $p_i$, arm $a_j$ pair, by the divergence decomposition (Lemma~\ref{lem:divergence_decompose}), we have
	\begin{align*}
		&D(\sP_{\boldsymbol{\nu}, \boldsymbol{\pi}}, \sP_{\boldsymbol{\nu}', \boldsymbol{\pi}}) \\
		=& D(\boldsymbol{\nu}_{i, j}, \boldsymbol{\nu}'_{i, j}) \mathbb{E}_{\boldsymbol{\nu}, \boldsymbol{\pi}} \left[N_j^{(i)}(T)\right] \\
		\leq& \left(\epsilon + D_{\inf} \left(\boldsymbol{\nu}_{i, j}, \mU(i, \mu^*(p_i)), \gP\right)\right) \mathbb{E}_{\boldsymbol{\nu}, \boldsymbol{\pi}} \left[N_j^{(i)}(T)\right],
	\end{align*}
	where the last inequality comes from the construction of $\boldsymbol{\nu}'_{i, j}$.
	
	Substituting the above three relations in Eq.(\ref{eq:BH_inequality}) we obtain for any $\epsilon > 0$:
	\begin{align*}
		&\left(\epsilon + D_{\inf} \left(\boldsymbol{\nu}_{i, j}, \mU(i, \mu^*(p_i)), \gP\right)\right) \mathbb{E}_{\boldsymbol{\nu}, \boldsymbol{\pi}} \left[N_j^{(i)}(T)\right] \\
		\geq& \log \left(\frac{1}{2 \left(\sP_{\boldsymbol{\nu}, \boldsymbol{\pi}}(A) + \sP_{\boldsymbol{\nu}', \boldsymbol{\pi}}(A^c)\right)}\right) \\
		\geq& \log \left(\frac{T \min \left\{\mU'(i, j) - \mU(i, \mu^*(p_i)), \Delta_j^{(i)}\right\}}{4 \left(Reg_i(T; \boldsymbol{\nu}, \boldsymbol{\pi}) + Reg_i(T; \boldsymbol{\nu}', \boldsymbol{\pi})\right)}\right),
	\end{align*}
	where the last inequality holds since the policy $\boldsymbol{\pi}$ is assumed to be universally consistent. Taking $\epsilon$ and $T$ to the limit, we have
	\begin{align*}
		&\lim_{\epsilon \to 0} \liminf_{T \to \infty} \frac{\mathbb{E}_{\boldsymbol{\nu}, \boldsymbol{\pi}} \left[N_j^{(i)}(T)\right]}{\log T} \\
		\geq& \lim_{\epsilon \to 0} \frac{1}{\epsilon + D_{\inf} \left(\boldsymbol{\nu}_{i, j}, \mU(i, \mu^*(p_i)), \gP\right)} \\
		=& \frac{1}{D_{\inf} \left(\boldsymbol{\nu}_{i, j}, \mU(i, \mu^*(p_i)), \gP\right)}
	\end{align*}
	The above bound is true for any uniformly consistent policy $\boldsymbol{\pi}$, and for player $p_i$, and arm $a_j$ for $i, j \in [N]$. We use the regret decomposition (Lemma~\ref{lem:regret_decompose}) to obtain the final asymptotic regret lower bound for any player $p_i$ as
	\begin{align*}
		&\liminf_{T \to \infty} \frac{Reg_i(T; \boldsymbol{\nu}, \pi)}{\log T} \\
		\geq& \max \Biggl\{\sum_{i' \neq i} \frac{\Delta_{\min}^{(i)}}{D_{\inf}(\boldsymbol{\nu}_{i', \mu^*(p_i)}, \mU(i', \mu^*(p_{i'})), \gP)}, 
		\sum_{a_j \neq \mu^*(p_i)} \frac{\Delta_{j}^{(i)}}{D_{\inf}(\boldsymbol{\nu}_{i, j}, \mU(i, \mu^*(p_i)), \gP)}\Biggr\}.
	\end{align*}
\end{proof}

\subsection{Proof of Corollary~\ref{coro:lower_bound}}\label{subsec:proof_coro_lower_bound}
The above corollary follows readily from Theorem~\ref{thm:lower_bound}. Consider the following STTCB instance, let the optimal arm for any player $p_i$ be $a_{i + 1}$, then all players together form a top trading cycle. For a given index $i$, let for any player $p_{i'}$, $i' \neq i$, set $\mU(i', i' + 1) = \frac{1}{2}$ while for any other arm $a_j$, $j \neq i' + 1$, set $\mU(i', j) = \frac{1}{2} - \Delta$, where $\Delta > 0$ is small enough. Also, let player $p_i$ has the arm mean for $a_{i + 1}$ be $\mU(i, i + 1) = \frac{1}{2}$ and for any other arms $a_j$, $j \neq i + 1$ set $\mU(i, j) = \frac{1}{4}$. For $\gP$ the class of Bernoulli rewards, we have
\begin{align*}
	D_{\inf} \left(\boldsymbol{\nu}_{i', \mu^*(p_i)}, \mU(i', \mu^*(p_i)), \gP\right) \leq \frac{\Delta^2}{4}, \forall i' \neq i,
\end{align*}
and $\Delta_{\min}^{(i)} = \frac{1}{4}$. Therefore, the regret for player $p_i$ is lower bounded by $\frac{(N - 1) \log T}{16 \Delta^2}$.

\section{Technical Lemmas}\label{sec:tech_lemma}
\begin{lem}[Corollary 5.5 in \citet{lattimore2020bandit}]\label{lem:concentration}
	Assume that $X_1, X_2, \cdots, X_n$ are independent $\sigma$-subgaussian random variables centered around $\mu$. Then for any $\varepsilon > 0$, 
	\begin{align*}
		&\sP \left(\frac{1}{n} \sum_{i = 1}^n X_i \geq \mu + \varepsilon\right) \leq \exp \left(- \frac{n \varepsilon^2}{2 \sigma^2}\right), \\
		&\sP \left(\frac{1}{n} \sum_{i = 1}^n X_i \leq \mu - \varepsilon\right) \leq \exp \left(- \frac{n \varepsilon^2}{2 \sigma^2}\right).
	\end{align*}
\end{lem}

\begin{lem}[Bretagnolle-Huber Inequality, Theorem 14.2~\citep{lattimore2020bandit}]\label{lem:BH_inequality}
	Let $P$ and $Q$ be probability measures on the same measurable space $(\Omega, \gF)$, and let $A \in \gF$ be an arbitrary event. Then, 
	\begin{align*}
		P(A) + Q(A^c) \geq \frac{1}{2} \exp \left(-D(P, Q)\right),
	\end{align*}
	where $A^c = \Omega \backslash A$ is the complement of $A$.
\end{lem}

\end{document}